\documentclass[11pt]{article}

\usepackage{times,graphics,psfrag,amsthm,amsmath,amssymb}

\newtheorem{theorem}{Theorem}
\newtheorem{lemma}{Lemma}

\newcommand{\RR}{\mathbb{R}}
\newcommand{\dist}{\mathrm{dist}}

\begin{document}

\title{\bf Spreading grid cells}

\author{Minghui Jiang \quad Pedro J. Tejada\medskip\\
Department of Computer Science, Utah State University, Logan, UT 84322, USA\medskip\\
\texttt{mjiang@cc.usu.edu} \quad \texttt{p.tejada@aggiemail.usu.edu}}

\maketitle

\begin{abstract}
Let $S$ be a set of $n^2$ symbols.
Let $A$ be an $n\times n$ square grid
with each cell labeled by a distinct symbol in $S$.
Let $B$ be another $n\times n$ square grid,
also with each cell labeled by a distinct symbol in $S$.
Then each symbol in $S$ labels two cells, one in $A$ and one in $B$.
Define the \emph{combined distance} between two symbols in $S$
as the distance between the two cells in $A$
plus the distance between the two cells in $B$
that are labeled by the two symbols.
Bel\'en Palop asked the following question
at the open problems session of CCCG 2009:
How to arrange the symbols in the two grids
such that
the minimum combined distance between any two symbols is maximized?
In this paper, we give a partial answer to Bel\'en Palop's question.

Define $c_p(n) = \max_{A,B}\min_{s,t \in S}
	\{ \dist_p(A,s,t) + \dist_p(B,s,t) \}$,
where $A$ and $B$ range over all pairs of $n\times n$ square grids
labeled by the same set $S$ of $n^2$ distinct symbols,
and where $\dist_p(A,s,t)$ and $\dist_p(B,s,t)$ are the $L_p$ distances
between the cells in $A$ and in $B$, respectively,
that are labeled by the two symbols $s$ and $t$.
We present asymptotically optimal bounds
$c_p(n) = \Theta(\sqrt{n})$ for all $p=1,2,\ldots,\infty$.
The bounds also hold for generalizations to $d$-dimensional grids
for any constant $d \ge 2$.
Our proof yields
a simple linear-time constant-factor approximation algorithm
for maximizing the minimum combined distance
between any two symbols in two grids.
\end{abstract}


\section{Introduction}

Let $S$ be a set of $n^2$ symbols.
Let $A$ be an $n\times n$ square grid
with each cell labeled by a distinct symbol in $S$.
Let $B$ be another $n\times n$ square grid,
also with each cell labeled by a distinct symbol in $S$.
Then each symbol in $S$ labels two cells, one in $A$ and one in $B$.
Define the \emph{combined distance} between two symbols in $S$
as the distance between the two cells in $A$
plus the distance between the two cells in $B$
that are labeled by the two symbols.
Bel\'en Palop asked the following question
at the open problems session of CCCG 2009~\cite{DO09}:
\begin{quote}
How to arrange the symbols in the two grids
such that
the minimum combined distance between any two symbols is maximized?
\end{quote}

\begin{figure}[htbp]
\psfrag{1}{$1$}
\psfrag{2}{$2$}
\psfrag{3}{$3$}
\psfrag{4}{$4$}
\psfrag{5}{$5$}
\psfrag{6}{$6$}
\psfrag{7}{$7$}
\psfrag{8}{$8$}
\psfrag{9}{$9$}
\psfrag{A}{$A$}
\psfrag{B}{$B$}
\centering
\includegraphics{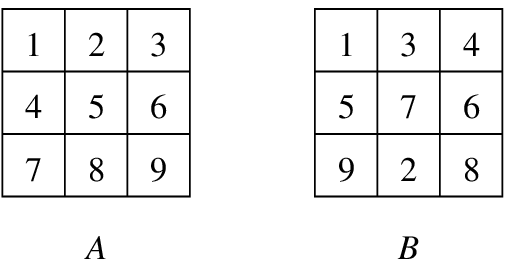}
\caption{Two $3\times 3$ grids $A$ and $B$ labeled by $S=\{1,2,3,4,5,6,7,8,9\}$.
Given the grid $A$, the grid $B$ is one of $840$ solutions
found by a computer program such that the combined $L_1$ distance
between any two symbols in the two grids is at least $3$.}
\label{fig:3x3}
\end{figure}

In the original setting of this question as posed by
Bel\'en Palop,
the two grids $A$ and $B$ are axis-parallel,
each grid cell is a unit square,
and the distance between two cells is the $L_1$ distance between
the cell centers.
Thus the distance between two cells sharing an edge is $1$,
and the distance between two cells sharing only a vertex is $2$.
We refer to Figure~\ref{fig:3x3} for an example.
Note that the question is also interesting for the other norms $L_p$,
$p=2,\ldots,\infty$, in particular, $L_\infty$.

In this paper, we give a partial answer to Bel\'en Palop's question.
To be precise, let $n \ge 2$, and define
\begin{equation}\label{eq:c}
c_p(n) = \max_{A,B}\min_{s,t \in S}
	\big\{ \dist_p(A,s,t) + \dist_p(B,s,t) \big\},
\end{equation}
where $A$ and $B$ range over all pairs of $n\times n$ square grids
labeled by the same set $S$ of $n^2$ distinct symbols,
and where $\dist_p(A,s,t)$ and $\dist_p(B,s,t)$ are the $L_p$ distances
between (the centers of) the cells in $A$ and in $B$, respectively,
that are labeled by the two symbols $s$ and $t$.
Our main result is the following theorem:

\begin{theorem}\label{T1}
For any integer $n \ge 2$,
$2\big\lfloor \sqrt{n/3} \,\big\rfloor \le c_\infty(n)
	\le \big\lceil \sqrt{n-1} \,\big\rceil
		+ \big\lfloor \sqrt{n-1} \,\big\rfloor$.
Consequently, for any integers $n \ge 2$ and $p \ge 1$,
$2\big\lfloor \sqrt{n/3} \,\big\rfloor \le c_p(n)
	\le 2^{1/p}\big( \big\lceil \sqrt{n-1} \,\big\rceil
		+ \big\lfloor \sqrt{n-1} \,\big\rfloor \big)$.
\end{theorem}

Our bounds on the minimum combined distance can be generalized to
$d$-dimensional grids for any integer $d \ge 2$.
Define $c_p^d(n)$ analogous to \eqref{eq:c} except that
$A$ and $B$ range over all pairs of $n\times\cdots\times n$
hypercubic grids labeled by the same set $S$ of $n^d$ distinct symbols.
Then $c_p^2(n) = c_p(n)$ for all $p=1,2,\ldots,\infty$.
Theorem~\ref{T2} in the following
is a straightforward extension to Theorem~\ref{T1}:

\begin{theorem}\label{T2}
For any integers $d \ge 2$ and $n \ge 2$,
$2\big\lfloor \sqrt{n/3} \,\big\rfloor \le c_\infty^d(n)
	\le \big\lceil \sqrt{n-1} \,\big\rceil
		+ \big\lfloor \sqrt{n-1} \,\big\rfloor$.
Consequently, for any integers $d \ge 2$, $n \ge 2$, and $p \ge 1$,
$2\big\lfloor \sqrt{n/3} \,\big\rfloor \le c_p^d(n)
	\le d^{1/p}\big( \big\lceil \sqrt{n-1} \,\big\rceil
		+ \big\lfloor \sqrt{n-1} \,\big\rfloor \big)$.
\end{theorem}

Our proof for the lower bound is constructive and,
in conjunction with the upper bound,
yields a simple linear-time constant-factor approximation algorithm
for the optimization problem of maximizing the minimum combined distance
between any two symbols in two grids.

\section{Lower Bound}

In this section, we prove the lower bound
$c_\infty(n) \ge 2\big\lfloor \sqrt{n/3} \,\big\rfloor$
in Theorem~\ref{T1}.
For convenience, let
$$
S=\{ (x,y) \mid 0 \le x,y \le n-1 \}
$$
be the set of center coordinates of the grid cells of $A$,
and label each cell of $A$ by its center coordinates.
Let
$$
C=\{ (i,j) \mid 0 \le i,j \le k-1 \}
$$
be a set of $k^2$ colors,
where $k = \Theta(\sqrt n)$ is a positive integer to be specified.
To prove the lower bound,
we will construct $B$ from $A$ by moving cells in the same grid
such that the combined distance between any two symbols in $A$ and $B$
is $\Omega(k)$.

\subsection{A Special Case}

We first consider the special case that $n = k^2$ for some integer $k \ge 2$.
Assign a color $(i,j)$ to each cell $(x,y)$ such that
$i = x \bmod k$ and $j = y \bmod k$.
To transform $A$ into $B$,
we move
each cell $(x,y)$ of color $(i,j)$
to a cell $(x',y')$ of the same color $(i,j)$ such that
\begin{align}
x' &= \left\{
\begin{array}{ll}
x + i k,
	& \textup{if }\; x + i k \le n - 1
\\
x + i k - k^2,
	& \textup{if }\; x + i k > n - 1
\end{array}
\right.
\\
y' &= \left\{
\begin{array}{ll}
y + j k,
	& \textup{if }\; y + j k \le n - 1
\\
y + j k - k^2,
	& \textup{if }\; y + j k > n - 1
\end{array}
\right.
\end{align}
Then each cell in $A$ is moved to a distinct cell in $B$.
The cells of color $(0,0)$ remain at the same positions in the grid.
We refer to Figure~\ref{fig:9x9} for an example.

\begin{figure}[htbp]
\centering
\hspace*{\stretch1}
\begin{minipage}[t]{0.45\linewidth}
\centering
\resizebox{\linewidth}{!}{\includegraphics{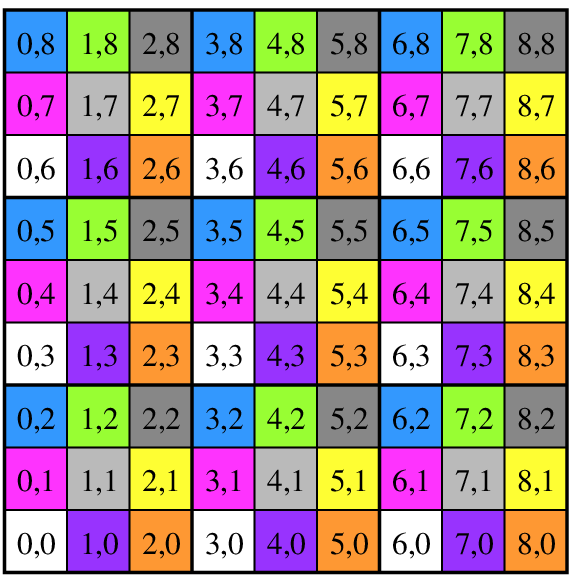}}\\\smallskip$A$
\end{minipage}
\hspace*{\stretch1}
\begin{minipage}[t]{0.45\linewidth}
\centering
\resizebox{\linewidth}{!}{\includegraphics{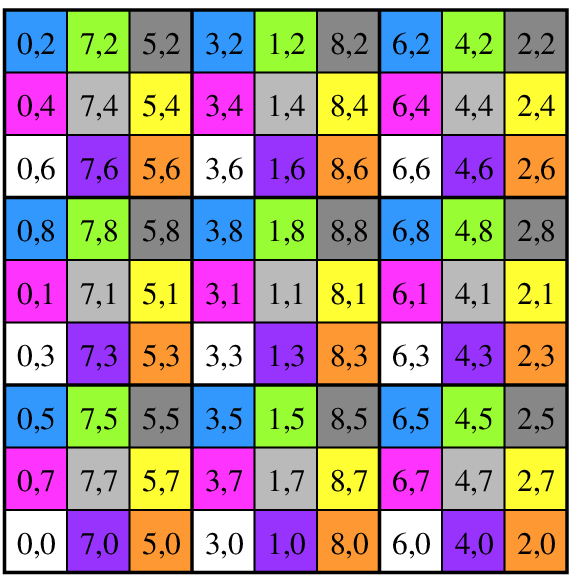}}\\\smallskip$B$
\end{minipage}
\hspace*{\stretch1}
\caption{Two grids $A$ and $B$ for $n = 9$ and $k = 3$.}
\label{fig:9x9}
\end{figure}

Consider any
two cells $(x_1, y_1)$ and $(x_2, y_2)$ in $A$
that are moved to
two cells $(x'_1, y'_1)$ and $(x'_2, y'_2)$ in $B$.
The combined $L_\infty$ distance between the corresponding two symbols
$(x_1, y_1)$ and $(x_2, y_2)$ is
$$
\max\{ |x_2 - x_1|, |y_2 - y_1| \} + \max\{ |x'_2 - x'_1| + |y'_2 - y'_1| \}.
$$
We will show that this combined distance is at least $k$.
Let
$$
i_1 = x_1 \bmod k,
\quad
i_2 = x_2 \bmod k,
\quad
j_1 = y_1 \bmod k,
\quad
j_2 = y_2 \bmod k.
$$
If $(i_1,j_1) = (i_2,j_2)$,
then the combined distance is at least $2k$
because the $L_\infty$ distance between any two cells of the same color
is at least $k$.
It remains to show that the combined distance is at least $k$
even if $(i_1,j_1) \neq (i_2,j_2)$.
Assume without loss of generality that $i_1 \neq i_2$.
It suffices to show that $|x_2 - x_1| + |x'_2 - x'_1| \ge k$.
Consider four cases:

\begin{enumerate}

\item
$x'_1 = x_1 + i_1 k$
and
$x'_2 = x_2 + i_2 k$.
\begin{align*}
|x'_2 - x'_1| + |x_2 - x_1|
&= | (x_2 + i_2 k) - (x_1 + i_1 k) | + |x_2 - x_1|
\\
&\ge | (x_2 + i_2 k) - (x_1 + i_1 k) - (x_2 - x_1) |
\\
&= |i_2 - i_1| \cdot k.
\end{align*}

\item
$x'_1 = x_1 + i_1 k - k^2$
and
$x'_2 = x_2 + i_2 k - k^2$.
\begin{align*}
|x'_2 - x'_1| + |x_2 - x_1|
&= | (x_2 + i_2 k - k^2) - (x_1 + i_1 k - k^2) | + |x_2 - x_1|
\\
&\ge | (x_2 + i_2 k - k^2) - (x_1 + i_1 k - k^2) - (x_2 - x_1) |
\\
&= |i_2 - i_1| \cdot k.
\end{align*}

\item
$x'_1 = x_1 + i_1 k$
and
$x'_2 = x_2 + i_2 k - k^2$.
\begin{align*}
|x'_2 - x'_1| + |x_2 - x_1|
&= | (x_2 + i_2 k - k^2) - (x_1 + i_1 k) | + |x_2 - x_1|
\\
&\ge | (x_2 + i_2 k - k^2) - (x_1 + i_1 k) - (x_2 - x_1) |
\\
&= |i_2 - i_1 - k| \cdot k.
\end{align*}

\item
$x'_1 = x_1 + i_1 k - k^2$
and
$x'_2 = x_2 + i_2 k$.
\begin{align*}
|x'_2 - x'_1| + |x_2 - x_1|
&= | (x_2 + i_2 k) - (x_1 + i_1 k - k^2) | + |x_2 - x_1|
\\
&\ge | (x_2 + i_2 k) - (x_1 + i_1 k - k^2) - (x_2 - x_1) |
\\
&= |i_2 - i_1 + k| \cdot k.
\end{align*}

\end{enumerate}

Recall that $0 \le i_1,i_2 \le k-1$ and $i_1 \neq i_2$.
Thus $1 \le |i_2 - i_1| \le k-1$.
This implies that
the two values $|i_2 - i_1 - k|$ and $|i_2 - i_1 + k|$
are both at least $1$.
In summary,
we have
$|x'_2 - x'_1| + |x_2 - x_1| \ge k$
in all four cases.

\subsection{The General Case}

Let $k$ be the largest integer such that
$3k \le \lceil n/k \rceil$;
we will show later that
$\big\lfloor \sqrt{n/3} \,\big\rfloor
\le k \le
\big\lceil \sqrt{n/3} \,\big\rceil$.
Again assign a color $(i,j)$ to each cell $(x,y)$ such that
$i = x \bmod k$ and $j = y \bmod k$.
To transform $A$ into $B$,
we move
each cell $(x,y)$ of color $(i,j)$
to a cell $(x',y')$ of the same color $(i,j)$ such that
\begin{align}
x' &= \left\{
\begin{array}{ll}
x + 3 i k,
	& \textup{if }\; x + 3 i k \le n - 1
\\
x + 3 i k - \lceil n/k \rceil k,
	& \textup{if }\; x + 3 i k > n - 1 \;\textup{ and }\; i \le (n - 1) \bmod k
\\
x + 3 i k - \lceil n/k \rceil k + k,
	& \textup{if }\; x + 3 i k > n - 1 \;\textup{ and }\; i > (n - 1) \bmod k
\end{array}
\right.
\\
y' &= \left\{
\begin{array}{ll}
y + 3 j k,
	& \textup{if }\; y + 3 j k \le n - 1
\\
y + 3 j k - \lceil n/k \rceil k,
	& \textup{if }\; y + 3 j k > n - 1 \;\textup{ and }\; j \le (n - 1) \bmod k
\\
y + 3 j k - \lceil n/k \rceil k + k,
	& \textup{if }\; y + 3 j k > n - 1 \;\textup{ and }\; j > (n - 1) \bmod k
\end{array}
\right.
\end{align}
Then each cell in $A$ is moved to a distinct cell in $B$.
The cells of color $(0,0)$ remain at the same positions in the grid.
We refer to Figure~\ref{fig:17x17} for an example.

\begin{figure}[htbp]
\centering
\resizebox{0.63\linewidth}{!}{\includegraphics{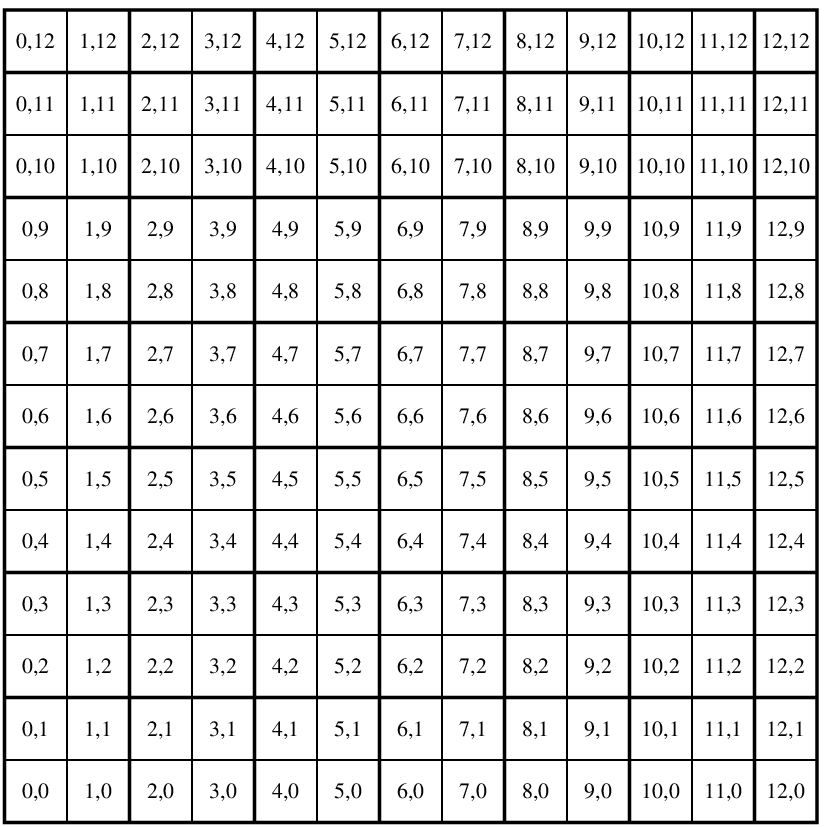}} $A$%
\bigskip\\
\resizebox{0.63\linewidth}{!}{\includegraphics{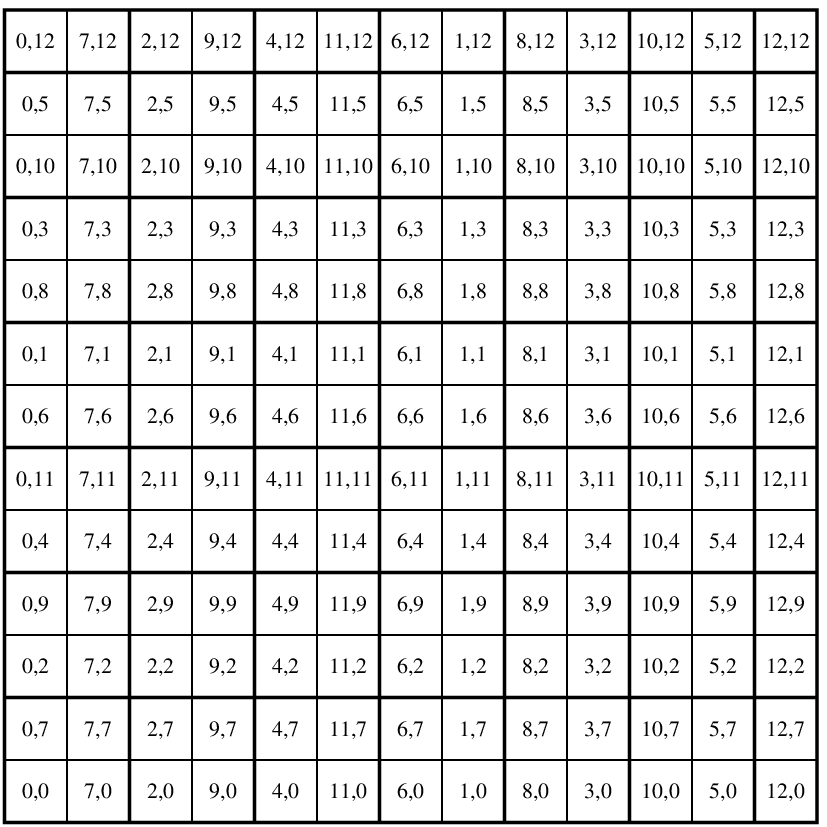}} $B$
\caption{Two grids $A$ and $B$ for $n = 13$ and $k = 2$.}
\label{fig:17x17}
\end{figure}

Consider any
two cells $(x_1, y_1)$ and $(x_2, y_2)$ in $A$
that are moved to
two cells $(x'_1, y'_1)$ and $(x'_2, y'_2)$ in $B$.
The combined $L_\infty$ distance between the corresponding two symbols
$(x_1, y_1)$ and $(x_2, y_2)$ is
$$
\max\{ |x_2 - x_1|, |y_2 - y_1| \} + \max\{ |x'_2 - x'_1| + |y'_2 - y'_1| \}.
$$
We will show that this combined distance is at least $2k$.
Let
$$
i_1 = x_1 \bmod k,
\quad
i_2 = x_2 \bmod k,
\quad
j_1 = y_1 \bmod k,
\quad
j_2 = y_2 \bmod k.
$$
If $(i_1,j_1) = (i_2,j_2)$,
then the combined distance is at least $2k$
because the $L_\infty$ distance between any two cells of the same color
is at least $k$.
It remains to show that the combined distance is at least $k$
even if $(i_1,j_1) \neq (i_2,j_2)$.
Assume without loss of generality that $i_1 \neq i_2$.
It suffices to show that $|x_2 - x_1| + |x'_2 - x'_1| \ge k$.
Consider nine cases:

\begin{enumerate}

\item
$x'_1 = x_1 + 3 i_1 k$
and
$x'_2 = x_2 + 3 i_2 k$.
\begin{align*}
|x'_2 - x'_1| + |x_2 - x_1|
&= | (x_2 + 3 i_2 k) - (x_1 + 3i_1 k) | + |x_2 - x_1|
\\
&\ge | (x_2 + 3 i_2 k) - (x_1 + 3i_1 k) - (x_2 - x_1) |
\\
&= |3(i_2 - i_1)| \cdot k.
\end{align*}

\item
$x'_1 = x_1 + 3 i_1 k - \lceil n/k \rceil k$
and
$x'_2 = x_2 + 3 i_2 k - \lceil n/k \rceil k$.
\begin{align*}
|x'_2 - x'_1| + |x_2 - x_1|
&= | (x_2 + 3i_2 k - \lceil n/k \rceil k)
	- (x_1 + 3i_1 k - \lceil n/k \rceil k) | + |x_2 - x_1|
\\
&\ge | (x_2 + 3i_2 k - \lceil n/k \rceil k)
	- (x_1 + 3i_1 k - \lceil n/k \rceil k) - (x_2 - x_1) |
\\
&= |3(i_2 - i_1)| \cdot k.
\end{align*}

\item
$x'_1 = x_1 + 3 i_1 k$
and
$x'_2 = x_2 + 3 i_2 k - \lceil n/k \rceil k$.
\begin{align*}
|x'_2 - x'_1| + |x_2 - x_1|
&= | (x_2 + 3i_2 k - \lceil n/k \rceil k) - (x_1 + 3i_1 k) | + |x_2 - x_1|
\\
&\ge | (x_2 + 3i_2 k - \lceil n/k \rceil k) - (x_1 + 3i_1 k) - (x_2 - x_1) |
\\
&= |3(i_2 - i_1) - \lceil n/k \rceil| \cdot k.
\end{align*}

\item
$x'_1 = x_1 + 3 i_1 k - \lceil n/k \rceil k$
and
$x'_2 = x_2 + 3 i_2 k$.
\begin{align*}
|x'_2 - x'_1| + |x_2 - x_1|
&= | (x_2 + 3i_2 k) - (x_1 + 3i_1 k - \lceil n/k \rceil k) | + |x_2 - x_1|
\\
&\ge | (x_2 + 3i_2 k) - (x_1 + 3i_1 k - \lceil n/k \rceil k) - (x_2 - x_1) |
\\
&= |3(i_2 - i_1) + \lceil n/k \rceil| \cdot k.
\end{align*}

\item
$x'_1 = x_1 + 3 i_1 k - \lceil n/k \rceil k + k$
and
$x'_2 = x_2 + 3 i_2 k - \lceil n/k \rceil k + k$.
\begin{align*}
|x'_2 - x'_1| + |x_2 - x_1|
&= | (x_2 + 3i_2 k - \lceil n/k \rceil k + k)
	- (x_1 + 3i_1 k - \lceil n/k \rceil k + k) | + |x_2 - x_1|
\\
&\ge | (x_2 + 3i_2 k - \lceil n/k \rceil k + k)
	- (x_1 + 3i_1 k - \lceil n/k \rceil k + k) - (x_2 - x_1) |
\\
&= |3(i_2 - i_1)| \cdot k.
\end{align*}

\item
$x'_1 = x_1 + 3 i_1 k$
and
$x'_2 = x_2 + 3 i_2 k - \lceil n/k \rceil k + k$.
\begin{align*}
|x'_2 - x'_1| + |x_2 - x_1|
&= | (x_2 + 3i_2 k - \lceil n/k \rceil k + k) - (x_1 + 3i_1 k) | + |x_2 - x_1|
\\
&\ge | (x_2 + 3i_2 k - \lceil n/k \rceil k + k) - (x_1 + 3i_1 k) - (x_2 - x_1) |
\\
&= |3(i_2 - i_1) - (\lceil n/k \rceil - 1)| \cdot k.
\end{align*}

\item
$x'_1 = x_1 + 3 i_1 k - \lceil n/k \rceil k + k$
and
$x'_2 = x_2 + 3 i_2 k$.
\begin{align*}
|x'_2 - x'_1| + |x_2 - x_1|
&= | (x_2 + 3i_2 k) - (x_1 + 3i_1 k - \lceil n/k \rceil k + k) | + |x_2 - x_1|
\\
&\ge | (x_2 + 3i_2 k) - (x_1 + 3i_1 k - \lceil n/k \rceil k + k) - (x_2 - x_1) |
\\
&= |3(i_2 - i_1) + (\lceil n/k \rceil - 1)| \cdot k.
\end{align*}

\item
$x'_1 = x_1 + 3 i_1 k - \lceil n/k \rceil k$
and
$x'_2 = x_2 + 3 i_2 k - \lceil n/k \rceil k + k$.
\begin{align*}
|x'_2 - x'_1| + |x_2 - x_1|
&= | (x_2 + 3i_2 k - \lceil n/k \rceil k + k)
	- (x_1 + 3i_1 k - \lceil n/k \rceil k) | + |x_2 - x_1|
\\
&\ge | (x_2 + 3i_2 k - \lceil n/k \rceil k + k)
	- (x_1 + 3i_1 k - \lceil n/k \rceil k) - (x_2 - x_1) |
\\
&= |3(i_2 - i_1) + 1| \cdot k.
\end{align*}

\item
$x'_1 = x_1 + 3 i_1 k - \lceil n/k \rceil k + k$
and
$x'_2 = x_2 + 3 i_2 k - \lceil n/k \rceil k$.
\begin{align*}
|x'_2 - x'_1| + |x_2 - x_1|
&= | (x_2 + 3i_2 k - \lceil n/k \rceil k)
	- (x_1 + 3i_1 k - \lceil n/k \rceil k + k) | + |x_2 - x_1|
\\
&\ge | (x_2 + 3i_2 k - \lceil n/k \rceil k)
	- (x_1 + 3i_1 k - \lceil n/k \rceil k + k) - (x_2 - x_1) |
\\
&= |3(i_2 - i_1) - 1| \cdot k.
\end{align*}

\end{enumerate}

Note that the first four cases here are similar to those for the special case
in the previous subsection.
Recall that $0 \le i_1,i_2 \le k-1$ and $i_1 \neq i_2$.
Thus $3 \le |3(i_2 - i_1)| \le 3(k-1)$.
Since $3(k-1) = 3k - 3 \le \lceil n/k \rceil - 3$ by our choice of $k$,
it follows that the two values
$|3(i_2 - i_1) - \lceil n/k \rceil|$
and
$|3(i_2 - i_1) + \lceil n/k \rceil|$
are both at least $3$.
Then
the four values
$|3(i_2 - i_1) - (\lceil n/k \rceil - 1)|$,
$|3(i_2 - i_1) + (\lceil n/k \rceil - 1)|$,
$|3(i_2 - i_1) + 1|$,
and
$|3(i_2 - i_1) - 1|$
are all at least $2$.
In summary, we have
$|x'_2 - x'_1| + |x_2 - x_1| \ge 2k$
in all nine cases.

The following lemma gives an estimate of $k$:

\begin{lemma}
Let $n$ be an integer such that $n \ge 2$.
Let $k$ be the largest integer such that $3k \le \lceil n/k \rceil$.
Then
$\big\lfloor \sqrt{n/3} \,\big\rfloor
\le k \le
\big\lceil \sqrt{n/3} \,\big\rceil$.
\end{lemma}

\begin{proof}
We have $k \ge \big\lfloor \sqrt{n/3} \,\big\rfloor$ because
$$
3\big\lfloor \sqrt{n/3} \,\big\rfloor
\le 3\sqrt{n/3}
= \frac{n}{\sqrt{n/3}}
\le \frac{n}{\big\lfloor \sqrt{n/3} \,\big\rfloor}
\le \left\lceil \frac{n}{\big\lfloor \sqrt{n/3} \,\big\rfloor} \right\rceil.
$$

On the other hand,
we have $k \le \lceil \sqrt{n/3} \rceil$ because
$$
3\big\lceil \sqrt{n/3} \,\big\rceil
\ge 3\sqrt{n/3}
= \frac{n}{\sqrt{n/3}}
\ge \frac{n}{\big\lceil \sqrt{n/3} \,\big\rceil}
\implies
3\big\lceil \sqrt{n/3} \,\big\rceil
\ge \left\lceil \frac{n}{\big\lceil \sqrt{n/3} \,\big\rceil} \right\rceil.
\qedhere
$$
\end{proof}

This completes the proof of the lower bound
$c_\infty(n) \ge 2\big\lfloor \sqrt{n/3} \,\big\rfloor$
in Theorem~\ref{T1}.

\section{Upper Bound}

In this section, we prove the upper bound
$c_\infty(n) \le \big\lceil \sqrt{n-1} \,\big\rceil
		+ \big\lfloor \sqrt{n-1} \,\big\rfloor$
in Theorem~\ref{T1}.
Let $A$ and $B$ be two arbitrary $n\times n$ square grids
labeled by the same set $S$ of $n^2$ symbols.
We will show that there are two symbols in $S$
such that the combined $L_\infty$ distance between them
in the two grids $A$ and $B$ is at most
$\big\lceil \sqrt{n-1} \,\big\rceil
		+ \big\lfloor \sqrt{n-1} \,\big\rfloor$.

\begin{figure}[htbp]
\centering
\begin{minipage}[t]{0.6\linewidth}
\resizebox{\linewidth}{!}{\includegraphics{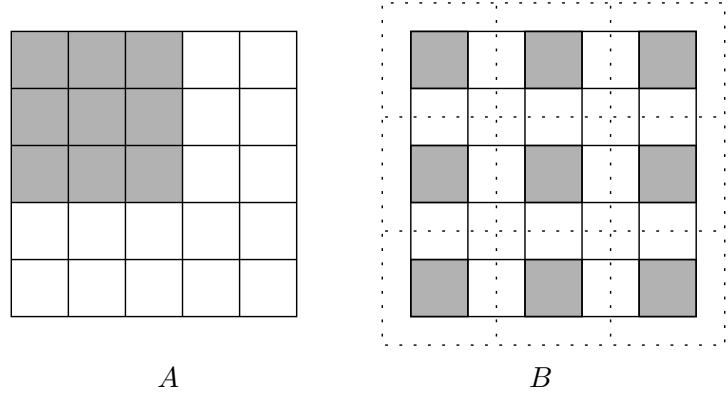}}\\
\hspace*{\stretch1}$A$\hspace*{\stretch2}$B$\hspace*{\stretch1}
\end{minipage}
\caption{Two grids $A$ and $B$ for $n = 5$ and $u = v = 2$.
The grid cells of $U$ and $V$ are shaded.
The $v\times v$ ($2\times 2$) squares that cover the cells in $V$ are dotted;
in this example, they tile the extended square of side $n - 1 + v = 6$.}
\label{fig:uv}
\end{figure}

Let $U$ be the set of cells in an arbitrary $(u+1)\times (u+1)$ sub-grid
of the $n\times n$ grid $A$,
where $u$ is an integer to be specified, $1 \le u \le n - 1$.
Then the $L_\infty$ distance between any two cells in $U$ is at most $u$.
Let $V$ be the set of cells in $B$ that are labeled by the same symbols
that label the cells in $U$.
Let $v$ be the minimum $L_\infty$ distance between any two cells in $V$.
For each cell in $V$, cover the cell by an axis-parallel square of side $v$
that is concentric with the cell.
Then these $v\times v$ squares are pairwise interior-disjoint,
and are all contained in an extended axis-parallel square of side
$n - 1 + v$
that is concentric with the grid $B$.
By an area argument, we have
\begin{equation}\label{eq:area}
(u+1)^2 \cdot v^2 \le (n - 1 + v)^2,
\end{equation}
which simplifies to
\begin{equation}\label{eq:uv}
uv \le n - 1.
\end{equation}
Now choose
$u = \big\lceil \sqrt{n-1} \,\big\rceil$.
It follows that
$v \le \big\lfloor \sqrt{n-1} \,\big\rfloor$.
Consider any two cells of $L_\infty$ distance $v$ in $B$.
The combined distance between the corresponding two symbols
is at most
$$
u + v \le \big\lceil \sqrt{n-1} \,\big\rceil
	+ \big\lfloor \sqrt{n-1} \,\big\rfloor.
$$
This completes the proof of the upper bound
$c_\infty(n) \le \big\lceil \sqrt{n-1} \,\big\rceil
		+ \big\lfloor \sqrt{n-1} \,\big\rfloor$
in Theorem~\ref{T1}.

\bigskip
Note that for the lower bound,
our construction is symmetric for all $d$ dimensions
and our case analysis is restricted to only one dimension.
Also note that for the upper bound,
the area argument in \eqref{eq:area} can be generalized to
a volume argument in higher dimensions, which still yields
the same inequality in \eqref{eq:uv}.
Thus we obtain the same bounds
$$
2 \big\lfloor \sqrt{n/3} \,\big\rfloor \le c_\infty^d(n)
	\le \big\lceil \sqrt{n-1} \,\big\rceil
		+ \big\lfloor \sqrt{n-1} \,\big\rfloor
$$
in Theorem~\ref{T2}.
The bounds on $c_p(n)$ and $c_p^d(n)$
in Theorem~\ref{T1} and Theorem~\ref{T2}
follow immediately because for any integer $p \ge 1$,
the $L_p$ distance between any two points in $\RR^d$
is at least the $L_\infty$ distance
and at most $d^{1/p}$ times the $L_\infty$ distance
between the two points.

\paragraph{Remark.}
After the submission of this manuscript,
the authors were informed by Joseph O'Rourke
that Vincent Pilaud, Nils Schweer, and Daria Schymura
had simultaneously and independently
obtained similar bounds $c_1(n) = \Theta(\sqrt n)$.

\end{document}